\documentclass[journal,onecolumn]{IEEEtran}
\usepackage{cite}
\usepackage{amsmath,amssymb,amsfonts, amsthm}
\usepackage{tabularx}
\usepackage{amsmath}
\usepackage{amsthm}
\usepackage{amssymb}
\usepackage{booktabs}
\usepackage{enumerate}
\usepackage{latexsym}
\usepackage{mathtools}
\usepackage{amsfonts}
\usepackage{setspace}
\usepackage{bm}
\usepackage{color}
\usepackage{soul}
\usepackage{verbatim}
\usepackage{makecell}
\usepackage{longtable}
\usepackage{orcidlink}
\usepackage{subcaption} 
\usepackage{caption}
\newtheorem{theorem}{Theorem}
\newtheorem{corollary}{Corollary}
\newcommand{\mc}[1]{\mathcal{#1}}

\newcommand{\F}{{\bf{F}}}

\newcommand{\rank}{\mathrm{rank}}

\newcommand{\bC}{{\mathbf{{C}}}}
\newcommand{\bF}{{\mathbf{{F}}}}

\usepackage{makecell}
\begin{document}

\title{Concatenated Sum-Rank Codes}

\author{Huimin Lao, Hao Chen, San Ling, and Yaqi Chen
\thanks{
Huimin Lao (huimin.lao@ntu.edu.sg) is with Strategic Centre for Research in Privacy-Preserving Technologies and Systems, Nanyang Technological University, Singapore.
San Ling (lingsan@ntu.edu.sg) is with School of Physical and Mathematical Sciences, Nanyang Technological University, Singapore.
San Ling is also with VinUniversity, Vinhomes Ocean Park, Gia Lam, Hanoi 100000, Vietnam (Email: ling.s@vinuni.edu.vn).
Hao Chen (chenhao@fudan.edu.cn) and Yaqi Chen (chenyq@stu.jnu.edu.cn) are with the College of Information Science and Technology/Cyber Security, Jinan University, Guangzhou, China.
This research is supported by the National Research Foundation, Singapore and Infocomm Media Development Authority under its Trust Tech Funding Initiative. 
Any opinions, findings and conclusions or recommendations expressed in this material are those of the author(s)
and do not reflect the views of National Research Foundation, Singapore and Infocomm Media Development Authority.
This research was supported by NSFC Grant 62032009.
}}


\maketitle

\begin{abstract}
Sum-rank codes have wide applications in multishot network coding, distributed storage and the construction of space-time codes. Asymptotically good sequences of linearized algebraic geometry sum-rank codes, exceeding the Gilbert-Varshamov-like bound, were constructed in a recent paper published in IEEE Trans. Inf. Theory by E. Berardini and X. Caruso.  
We call this bound the Tsfasman-Vl\u{a}du\c{t}-Zink-like bound. In this paper, we introduce the concatenation of a sum-rank code and a Hamming metric code. 
Then many sum-rank codes with good parameters, which are better than sum-rank BCH codes, are constructed simply and explicitly.
Moreover, we obtain an asymptotically good sequence of sum-rank codes exceeding the Tsfasman-Vl\u{a}du\c{t}-Zink-like bound and the Gilbert-Varshamov-like bound.
\end{abstract}


\section{Introduction}
The Hamming weight $wt_H({\bf a})$ of a vector ${\bf a}=(a_1, \ldots, a_n) \in {\bf F}_q^n$ is the cardinality  of its support $$supp({\bf a})=\{i:a_i \neq 0\}.$$ The Hamming distance $d_H({\bf a}, {\bf b})$ between two vectors ${\bf a}$ and ${\bf b}$ is defined to be $wt_H({\bf a}-{\bf b})$. 
For a code ${\bf C} \subseteq {\bf F}_q^n$, its minimum Hamming distance is $$d_H=\min_{{\bf a} \neq {\bf b}} \{d_H({\bf a}, {\bf b}): {\bf a}, {\bf b} \in {\bf C}\}.$$ 
It is well-known that the minimum Hamming distance of a linear code ${\bf C}$ is the minimum Hamming weight of its non-zero codewords. 
For the theory of error-correcting codes in the Hamming metric, the reader is referred to \cite{MScode,Lint,HP}.
A code ${\bf C}$ in ${\bf F}_q^n$ of  minimum distance $d_H$ and cardinality $M$ is called an $(n, M, d_H)_q$ code. 
The code rate is $R({\bf C})=\frac{\log_q M}{n},$ and the relative distance is $\delta({\bf C})=\frac{d_H}{n}.$ 
The Gilbert-Varshamov bound asserts that for the positive number $\delta$ satisfying $0<\delta<1-\frac{1}{q}$, 
there exists a sequence of $q$-ary codes ${\bf C}_i \subseteq {\bf F}_q^{n_i}$, where $n_i$ goes to the infinity, 
such that $\lim\limits_i \delta({\bf C}_i) \geq \delta$, and $$R =\lim\limits_i R({\bf C}_i) \geq 1-\delta-H_q(\delta),$$ 
see \cite[Chapter 2]{HP} and \cite[Chapter 5]{Lint}. Here $H_q$ is the $q$-ary entropy function.
If ${\bf C}$ is a linear code with dimension $k$, we denote it by $[n, k, d_H]_q$.
For an $[n, k, d_H]_q$ linear code, the Singleton bound asserts that $d_H \leq n-k+1$. When the equality holds, this code is called a maximum distance separable (MDS) code.
Reed-Solomon codes are well-known MDS codes \cite{HP}.

Let ${\bf F}_q^{(n,m)}$ be the set of all $ n \times m$ matrices over $\bF_q$, where $n \leq m$.
The rank-metric distance on the space ${\bf F}_q^{(n, m)}$ is defined by  $d_r(A,B)= \rank(A-B)$. The minimum rank distance of a code ${\bf C} \subset {\bf F}_q^{(n,m)}$  is $$d_r({\bf C})=\min_{A\neq B} \{d_r(A,B): A, B \in {\bf C} \}.$$  
For a code ${\bf C}$ in ${\bf F}_q^{(n, m)}$ with the minimum rank distance $d_r({\bf C}) \geq d$, 
the Singleton-like bound asserts that the number of codewords in ${\bf C}$ is upper bounded by $q^{m(n-d+1)}$  \cite{Gabidulin}.  
A code satisfying the equality is called a maximum rank distance (MRD) code. 
The Gabidulin code in ${\bf F}_q^{(n, n)}$ consists of ${\bf F}_q$-linear mappings on ${\bf F}_q^n \cong {\bf F}_{q^n}$ defined by $q$-polynomials $a_0x+a_1x^q+\cdots+a_ix^{q^i}+\cdots+a_tx^{q^t}$, 
where $a_t,\ldots,a_0$ are arbitrary elements in ${\bf F}_{q^n}$ \cite{Gabidulin}. 
Then the minimum rank distance of the Gabidulin code is at least $n-t$ since each non-zero $q$-polynomial above has at most $q^t$ roots in ${\bf F}_{q^n}$. 
There are $q^{n(t+1)}$ such $q$-polynomials. Hence the size of the Gabidulin code is $q^{n(t+1)}$ and it is an MRD code.

The sum-rank metric is a generalization of the Hamming metric and the rank metric. Sum-rank codes were introduced in \cite{NU} for their applications in multishot network coding (see also  \cite{MK19,NPS}). Sum-rank codes were also used to construct space-time codes \cite{SK}, and codes for distributed storage \cite{CMST,MK,MP1}.

Let $n_i \leq m_i$ be $2t$ positive integers satisfying $m_1 \geq m_2 \geq \dots \geq m_t$. Let
\begin{align}
  \label{eq: space}
  {\bf F}_q^{(n_1, m_1), \ldots,(n_t, m_t)}=
{\bf F}_q^{(n_1,  m_1)} \bigoplus \cdots \bigoplus {\bf F}_q^{(n_t, m_t)}
\end{align}
be the set of all ${\bf x}=({\bf x}_1,\ldots,{\bf x}_t)$, where ${\bf x}_i \in {\bf F}_q^{(n_i, m_i)}$ and $i \in \{ 1,\ldots,t \}$. One can see that (\ref{eq: space}) is a linear space over ${\bf F}_q$ of dimension $\sum_{i=1}^t n_im_i$.
Let the sum-rank weight of ${\bf x}=({\bf x}_1, \ldots, {\bf x}_t) \in {\bf F}_q^{(n_1, m_1), \ldots,(n_t, m_t)}$ be defined as $$wt_{sr}({\bf x}) =\sum_{i=1}^t \rank({\bf x}_i).$$
The sum-rank distance between two elements ${\bf x}=({\bf x}_1, \ldots, {\bf x}_t)$ and ${\bf y}=({\bf y}_1, \ldots, {\bf y}_t)$ in ${\bf F}_q^{(n_1, m_1), \ldots,(n_t, m_t)}$ is $$d_{sr}({\bf x},{\bf y})=wt_{sr}({\bf x}-{\bf y})=\sum_{i=1}^t \rank({\bf x}_i-{\bf y}_i),$$
This is indeed a metric on ${\bf F}_q^{(n_1,m_1), \ldots,(n_t,m_t)}$.
When $t=1$, it is the rank metric. When $n_{1}=\ldots = n_{t} = m_{1} = \ldots = m_{t}=1$, the sum-rank metric is the Hamming metric.

A subset $\bC$ of ${\bf F}_q^{(n_1,m_1), \ldots,(n_t,m_t)}$ with the sum-rank metric is called
a \emph{sum-rank code} with {\it block length} $t$ and {\it matrix sizes} $n_1 \times m_1, \ldots, n_t\times m_t$.
If $n_{1}=\ldots = n_{t} = n $ and $m_{1} = \ldots = m_{t}=m$,
we say $\bC$ has matrix size $n \times m$. 
Its minimum sum-rank distance is defined as $$d_{sr}({\bf C})=\min_{{\bf x} \neq {\bf y}, {\bf x}, {\bf y} \in {\bf C}} d_{sr}({\bf x}, {\bf y}).$$  
The code rate of ${\bf C}$ is $R_{sr}({\bf {C}})=\frac{\log_q |{\bf C}|}{\sum_{i=1}^t n_im_i}$. 
The relative distance is $\delta_{sr}({\bf {C}})=\frac{d_{sr}({\bf C})}{\sum_{i=1}^{t} n_{i}}$.
When ${\bf C}$ is linear, it is called a {\it linear sum-rank code}.
For a sequence of sum-rank codes $\mc{C} = \{ {\bf{C}}_{i} \}_{i \geq 1}$ of matrix size $n \times m$ over $\F_{q}$,
where ${\bf{C}}_{i}$ has block length $t_{i}$, dimension $k_{i}$, and minimum sum-rank distance $d_{i}$,
we define the rate of $\mc{C}$ by $R_{sr}(\mc{C}) = \lim \limits_{i \rightarrow \infty} \frac{k_{i}}{nmt_{i}}$
and the relative sum-rank distance of $\mc{C}$ by $\delta_{sr}(\mc{C}) = \lim \limits_{i \rightarrow \infty} \frac{d_{i}}{nt_{i}}$.

Analogous to other metrics, one of the main problems for coding in the sum-rank metric is to construct good sum-rank codes with large sizes and large minimum sum-rank distances.
Bounds and constructions for sum-rank codes have been studied in recent years \cite{AKR, BGR1, BGR, Borello,CGLGMP,Lao,MP1,MP191,MK19,MP20,MP21,Neri,NSZ21,MP22}.
Decoding algorithms for sum-rank codes have been studied in \cite{Bartz, Bartz4, Bartz2, Puchinger, PR}.
Byrne \textit{et al.} \cite{BGR} gave fundamental properties and some bounds on sizes of sum-rank codes. For a nice survey of sum-rank codes and their applications, the reader is referred to \cite{MPK22}.
Some related works about sum-rank codes are given as follows.

\begin{enumerate}[(1)]
    \item Similar to distance-optimal codes in the Hamming metric \cite{Ding5}, 
    distance-optimal sum-rank codes can be defined as follows. Let $\bC$ be a sum-rank code with $M$ codewords and minimum sum-rank distance $d_{sr}$. If there is no sum-rank code with $M$ codewords and minimum sum-rank distance $d_{sr}+1$, then $\bC$ is a distance-optimal sum-rank code. 
    The codes achieving the sphere packing bound \cite{BGR} are distance-optimal.
    Sum-rank Hamming codes with matrix size $1 \times m$ constructed in \cite{MP191} are perfect, thus they are distance-optimal.
    However the minimum sum-rank distance of such codes is $3$. 
    Infinite families of distance-optimal cyclic sum-rank codes with minimum sum-rank distance $4$ were constructed in \cite{Cheng,Chen2}.
    \item The Singleton-like bound on sum-rank codes was proposed in \cite{MK,BGR}. A code attaining this bound is called a {maximum sum-rank distance (MSRD) code.} 
    For constructions of MSRD codes, the reader is referred to \cite{MP1,BGR,Neri,MP20,Chen,MP22,MPS24,MP24}.
    The first known MSRD codes are called linearized Reed-Solomon codes \cite{MP1}.
    However, the block length of such codes over ${\bf F}_q$ is up to $q-1$.
    \item  Cyclic-skew-cyclic (CSC) codes in the sum-rank metric, and one of its subfamilies called sum-rank BCH codes, were introduced in \cite{MP21}.
    Compared to linearized Reed-Solomon codes, CSC codes can be defined over arbitrarily small finite fields and their block lengths can be arbitrary.
    As a result, many sum-rank codes of parameters $n=m=2$ and $q=2$ were given in \cite{MP21}. In addition, a decoding algorithm for sum-rank BCH codes was given \cite{MP21}.
    
    Alfarano \textit{et al.} \cite{ALNWZ} proposed another family of CSC codes, which are obtained from the tensor product of a cyclic code in the Hamming metric and a skew-cyclic code in the rank metric.
    
    Chen \cite{Chen} proposed a construction of linear sum-rank codes by combining several linear codes in the Hamming metric and Gabidulin codes.
    The dimensions of many codes constructed in \cite{Chen} are larger than those in \cite{MP21} for a given prescribed minimum sum-rank distance.

    \item It was proved in \cite{OPB} that random linear sum-rank codes attain the Gilbert-Varshamov-like (GV-like) bound (see Theorem \ref{thm: GV-like bound}) with high probability.
    Berardini \textit{et al.} \cite[Theorem 4]{Berardini} gave a bound called the Tsfasman-Vl\u{a}du\c{t}-Zink-like (TVZ-like) bound.
    For some parameters, this bound exceeds the GV-like bound.
\end{enumerate}

Concatenated codes in the Hamming metric was first introduced in \cite{Forney}. Such codes are constructed by two linear codes over different finite fields.
Specifically, a concatenated code with parameters $[n_{1}n_{2}, k_{1}k_{2}, \geq d_{1}d_{2}]_{q}$ can be obtained from 
two codes with parameters $[n_{1}, k_{1}, d_{1}]_{q^{k_{2}}}$ and $[n_{2}, k_{2}, d_{2}]_{q}$ respectively.
A generalization of concatenated codes was given in \cite{Dumer}.

In this paper, we propose a construction of sum-rank codes obtained by concatenating a linear code in the Hamming metric with a linear sum-rank code.
The codes constructed are called concatenated sum-rank codes. The minimum sum-rank distance of such codes is analysed.
By the asymptotic Tsfasman-Vl\u{a}du\c{t}-Zink bound for codes in the Hamming metric \cite[Theorem 13.5.4]{HP},
we obtain the following result.

\begin{theorem}
	\label{thm: asymptotically_Introduction}
	Let $p$ be a prime power.
	Let $r$, $m$, $d$ and $t$ be integers.
	Suppose that there is a sum-rank code $\bf{C}$ over $\F_{p}$ of block length $t$ with matrix size $m \times m$, dimension $2r$, and minimum sum-rank distance $d$.
	There exists a sequence of sum-rank codes of matrix size $m \times m$ over $\F_{p}$ with relative sum-rank distance $\delta_{sr} \in (0,1)$
  such that the rate $R_{sr}$ satisfies \begin{align*}
		\frac{m^{2}t}{2r}R_{sr} + \frac{mt}{d}\delta_{sr} \geq 1 - \frac{1}{p^{r}-1}.
	\end{align*}
\end{theorem}
Note that the parameters of the sum-rank code $\bf{C}$ in Theorem \ref{thm: asymptotically_Introduction} can be arbitrary, except that the dimension needs to be even.
Hence many known constructions of sum-rank codes can be applied to Theorem \ref{thm: asymptotically_Introduction}.
For comparison, we consider two families of sum-rank codes to instantiate the code $\bf{C}$ in Theorem \ref{thm: asymptotically_Introduction}. One is sum-rank codes with minimum sum-rank distance $2$,
and the other is linearized Reed-Solomon codes. For these two cases,
it turns out that 
the rate of the sum-rank codes over $\F_{p}$ given in Theorem \ref{thm: asymptotically_Introduction}
exceeds the TVZ-like bound and the GV-like bound, for some ranges of relative sum-rank distance and some small prime powers $p$.

In addition, we give explicit constructions of concatenated sum-rank codes. We show that for some parameter regimes,
the dimensions of our constructed codes are better than those given in \cite{MP21} and \cite{Chen}.

The rest of this paper is organized as follows.
In Section \ref{sec: pre}, we recall some bounds for sum-rank codes.
In Section \ref{sec: aysm}, concatenated sum-rank codes are given. In addition, a sequence of asymptotically good concatenated sum-rank codes exceeding the TVZ-like bound and the GV-like bound are given.
In Section \ref{sec: construct}, some explicit concatenated sum-rank codes are constructed.
Section \ref{sec: conclusion} concludes this paper.

\section{Preliminaries}
\label{sec: pre}
In this section, we firstly recall the GV-like bound and the Singleton-like bound for sum-rank codes.
The Singleton-like bound given in \cite[Theorem III.2]{BGR} is as follows.
\begin{theorem}[\cite{BGR}]
    Let $\bC \subseteq {\bf F}_q^{(n_1,m_1), \ldots,(n_t,m_t)}$ be a sum-rank code. The minimum sum-rank distance $d_{sr}$ can be written uniquely in the form $$ d_{sr}=\sum_{i=1}^{j-1} n_i+\delta+1, $$ where $0 \leq \delta \leq n_j-1$.
	Then $$|{\bf C}| \leq q^{\sum_{i=j}^t n_im_i-m_j\delta}.$$
\end{theorem}

The GV-like bound for sum-rank codes with matrix size $n \times m$ was given in \cite[Theorem 7]{OPB} and \cite{Puchinger}.

\begin{theorem}[Asymptotic Gilbert-Varshamov-like bound \cite{OPB, Berardini}]
\label{thm: GV-like bound}
Let $q$ be a prime power. Let $n$, $m$, and $t$ be positive integers with $n \leq m$. Let $N=nt$.
For positive real numbers $R$ and $\delta$ with $2< \delta N $ and 
\begin{align}
  \label{eq: GV_exact}
  R & \leq \delta^2 \frac{n}{m}-\delta(1+\frac{n}{m}+2\frac{n}{Nm})+1+\frac{1}{N}+\frac{n}{Nm}+\frac{n}{N^2m} \notag \\
  & - \frac{\sum_{i=1}^{\delta N-1} \log_q(1+\frac{t-1}{i})+\log_q(\delta N-1)}{Nm}-\frac{\log_q (\gamma_q)}{nm},
\end{align}
there exists a linear sum-rank code in ${\bf F}_q^{(n, m), \ldots,(n, m)}$ of block length $t$, rate $R$ and 
relative sum-rank distance at least $\delta$. 
When $t$ goes to infinity and $n=m$,
the bound in (\ref{eq: GV_exact}) is asymptotically given by 
\begin{align}
\label{eq: asy_GV_bound}
R \leq & (\delta -1)^2-\frac{\delta}{m} \log _q\left(1+\frac{1}{\delta m}\right)-\frac{\log _q(1+\delta m)}{m^2} -\frac{\log_q\left(\gamma_q\right)}{m^2}+o(1),
\end{align}
where $\gamma_q=\prod_{i=1}^{\infty}\left(1-q^{-i}\right)^{-1}$ and a function $f(x)$ is $o(1)$ if and only if $\lim\limits_{x \rightarrow \infty} f(x) = 0$.
\end{theorem}

The Tsfasman-Vl\u{a}du\c{t}-Zink bound for codes in the Hamming metric \cite[Theorem 13.5.4]{HP} is given as follows.
\begin{theorem}[TVZ bound on the Hamming metric \cite{HP}]
\label{thm: TVZ bound}
Let $q=p^{2r}$, where $p$ is a prime power and $r$ is an integer.
There exists a sequence of codes $\mc{C} = \{ {\bf{C}}_{i} \}_{i \geq 1}$,
where ${\bf{C}}_{i}$ is an $[n_{i}, k_{i}, d_i]_{q}$ linear code and the lengths go to infinity,
such that the relative distance $\delta = \lim\limits_{i \rightarrow \infty} \frac{d_{i}}{n_{i}}$ and
the rate $R = \lim\limits_{i \rightarrow \infty} \frac{k_{i}}{n_{i}}$ 
satisfy $$R + \delta \geq 1 - \frac{1}{p^{r}-1}. $$

\end{theorem}

The Tsfasman-Vl\u{a}du\c{t}-Zink-like bound for the codes in the sum-rank metric \cite[Theorem 4]{Berardini} is given as follows.
\begin{theorem}[Tsfasman-Vl\u{a}du\c{t}-Zink-like bound \cite{Berardini}]
Let $p$ be a square prime power. For all real positive numbers $R,\delta \in (0,1)$ satisfying $$R\leq 1-\delta-\frac{2}{\sqrt{p}-1}+\frac{1}{m(\sqrt{p}-1)},$$ 
there exists a linear sum-rank code with matrix size $m \times m$ over $\F_{p}$ of rate at least $R$ and relative sum-rank distance at least $\delta$.
\end{theorem}


\section{Concatenated sum-rank codes and asymptotically good sum-rank codes exceeding the TVZ-like bound and the GV-like bound}
\label{sec: aysm}
In this section, we firstly give concatenated sum-rank codes. 

\begin{theorem}
	\label{thm: CSR}
	Let ${\bf C}_{1}$ be an $[n, k_{1}, d_{1}]_{q^{k_{2}}}$ linear code.
	Let ${\bf C}_{2}$ be a linear sum-rank code of block length $t$ and of matrix size $m \times m^{\prime}$ over ${\bf F}_{q}$, where $t$, $m$, and $m^{\prime}$ are positive integers.
	Let $k_{2}$ and $d_{2}$ be the dimension and the minimum distance of ${\bf C}_{2}$, respectively.
	Then we have a linear sum-rank code ${\bf C}$ of block length $nt$ and of matrix size $m \times m^{\prime}$ over ${\bf F}_{q}$.
	Moreover, the dimension of ${\bf C}$ is $k_{1} \cdot k_{2}$ and the minimum sum-rank distance of ${\bf C}$ is at least $d_{1} \cdot d_{2}$.
\end{theorem}
\begin{proof}
	Let $\pi$ be an ${\bf F}_{q}$-linear bijection from ${\bf F}_{q^{k_{2}}}$ to ${\bf F}_{q}^{k_{2}}$. 
	Let $\phi$ be a one-to-one mapping from ${\bf F}_{q}^{k_{2}}$ to ${\bf C}_{2}$, that is $\phi(\bm{c}) = \bm{c} \cdot G$ for any $\bm{c} \in {\bf F}_{q}^{k_{2}}$, where $G$ is a generator matrix of ${\bf C}_{2}$.
	Then we define the sum-rank code $$ {\bf C} = \{ (\phi(\pi(c_{1})), \ldots, \phi(\pi(c_{n}))): (c_{1}, \ldots, c_{n}) \in {\bf C}_{1}  \}. $$
	One can easily verify that ${\bf C}$ has dimension $k_{1} \cdot k_{2}$ and minimum sum-rank distance at least $d_{1} \cdot d_{2}$.
\end{proof}

Based on Theorem \ref{thm: CSR},
we obtain asymptotically good sum-rank codes by concatenating a linear code achieving 
the Tsfasman-Vl\u{a}du\c{t}-Zink bound and a sum-rank code.

\begin{theorem}
	\label{thm: asymptotically}
	Let $p$ be a prime power.
	Let $r$, $m$, $d$ and $t$ be integers.
	Suppose that there is a sum-rank code ${\bf C}$ over $\F_{p}$ of block length $t$ with matrix size $m \times m$, dimension $2r$, and minimum sum-rank distance $d$.
  There exists a sequence of sum-rank codes of matrix size $m \times m$ over $\F_{p}$ with relative sum-rank distance $\delta_{sr} \in (0,1)$
  such that the rate $R_{sr}$ satisfies \begin{align*}
		\frac{m^{2}t}{2r}R_{sr} + \frac{mt}{d}\delta_{sr} \geq 1 - \frac{1}{p^{r}-1}.
	\end{align*}
\end{theorem}

\begin{proof}
  By Theorem \ref{thm: TVZ bound}, there exists a sequence $\mc{C}=\{ {\bf {C}}_{i} \}_{i \geq 1}$ of linear codes,
  where ${\bf {C}}_{i}$ is an $[n_{i}, k_{i}, d_{i}]_{p^{2r}}$ code, such that as the lengths go to infinity,
  the relative distance $\delta(\mc{C}) = \lim\limits_{i \rightarrow \infty} \frac{d_{i}}{n_{i}}$ and 
  the rate $R(\mc{C}) = \lim\limits_{i \rightarrow \infty} \frac{k_{i}}{n_{i}}$ satisfy 
  \begin{align}
    \label{eq: rate_distance_eq}
    R(\mc{C}) \geq 1 - \delta(\mc{C}) - \frac{1}{p^{r}-1}.
  \end{align}
  By concatenating the $[n_{i}, k_{i}, d_{i}]_{p^{2r}}$ linear code ${\bf C}_{i}$ and ${\bf C}$,
	Theorem \ref{thm: CSR} shows that there exists 
  a sequence of linear sum-rank codes $\mc{C}^{\prime} = \{ {\bf C}_{i}^{\prime} \}_{i \geq 1}$,
  where ${\bf C}_{i}^{\prime}$ has block length $n_{i}t$ and matrix size $m \times m$.
  In addition, the dimension of ${\bf C}_{i}^{\prime}$ is $k_{i} \cdot 2r$ and minimum sum-rank distance $d_{r}({\bf C}_{i}^{\prime})$ is at least $d_{i} \cdot d$.
  The relative sum-rank distance of $\mc{C}^{\prime}$ is $\delta_{sr} = \lim \limits_{i \rightarrow \infty} \frac{d_{r}({\bf C}_{i}^{\prime})}{m n_{i} t}$
  and the rate of $\mc{C}^{\prime}$ is $R_{sr} = \lim \limits_{i \rightarrow \infty} \frac{k_{i} \cdot 2r}{m^{2}n_{i}t}$.
  Then 
  \begin{align*}
    \delta_{sr} & \geq \lim \limits_{i \rightarrow \infty} \frac{d_{i} \cdot d}{m n_{i} t} = \delta(\mc{C}) \cdot \frac{d}{mt}, \\
    R_{sr} &= R(\mc{C}) \cdot \frac{2r}{m^{2}t}.
  \end{align*}
  By (\ref{eq: rate_distance_eq}),
  we have $\frac{m^{2}t}{2r}R_{sr} + \frac{mt}{d}\delta_{sr} \geq 1 - \frac{1}{p^{r}-1}$.


\end{proof}

In the following subsections, we instantiate the codes $\bf {C}$ in Theorem \ref{thm: asymptotically} by a sum-rank code with minimum sum-rank distance $2$
and linearized Reed-Solomon codes, respectively.

\subsection{Concatenated codes from sum-rank codes with minimum sum-rank distance $2$}
When the sum-rank metric code ${\bf C}$ in Theorem \ref{thm: asymptotically} has minimum sum-rank distance $2$,
we have the following corollary.
\begin{theorem}
	\label{thm: asy_bound_d2}
	Let $p$ be a prime power.
	Let $t$ and $m$ be integers with $m^{2}(t-1)$ being even.
  Let $r = m^{2}(t-1)/2$.
  There exists a sequence of linear sum-rank codes of matrix size $m \times m$ over $\F_{p}$ with relative sum-rank distance $\delta_{sr} \in (0,1)$
  such that the rate $R_{sr}$ satisfies
  \begin{align*}
		\frac{t}{t-1}R_{sr} + \frac{mt}{2} \delta_{sr} \geq 1 - \frac{1}{p^{r}-1}.
	\end{align*}
\end{theorem}
\begin{proof}
    Let $O$ be the $m \times m$ zero matrix.
	Let \begin{align}
    \label{eq: code_C}
    {\bf C}= \{ (X_{1}, \dots, X_{t}) \in \Pi: \sum_{i=1}^{t} X_{i} = O \}
  \end{align}
   be a sum-rank code in $\Pi = \F_{p}^{(m, m),\dots,(m, m)}$ with block length $t$.
    Since $X_{1}, \dots, X_{t-1}$ can be arbitrary and they are matrices in $\F_{p}^{(m, m)}$, the dimension of ${\bf C}$ is $m^{2}(t-1)$.
    Let $U = (U_{1}, \dots, U_{t})$ and $V = (V_{1}, \dots, V_{t})$ be different codewords in $\bf{C}$.
    Assume that there is only one entry $U_{j}$ of $U$ and $V_{j}$ of $V$ such that $U_{j} \neq V_{j}$ for an integer $j \in \{ 1, \dots,t \}$.
    Then $U_{i} = V_{i}$ for all $i \in S:= \{ 1, \dots, t \} \setminus \{ j \}$.
    By (\ref{eq: code_C}), we have $U_{j} = -\sum_{i \in S} U_{i}$  and $V_{j} = -\sum_{i \in S} V_{i}$.
    Then $U_{j} = V_{j}$, which contradicts our assumption.
    Hence there are at least two integers $i_{1}$ and $j_{1}$ in $\{ 1, \dots,t \}$ such that $U_{i_{1}} \neq V_{i_{1}}$
    and $U_{j_{1}} \neq V_{j_{1}}$, and the minimum distance of ${\bf C}$
	is at least $2$. Clearly, there exists a codeword in $\bf{C}$ with sum-rank weight $2$.
    Hence the minimum sum-rank distance is $2$.
    By Theorem \ref{thm: asymptotically}, the conclusion follows.
\end{proof}

Let $m=2$ in Theorem \ref{thm: asy_bound_d2}.
We have the following corollary.
\begin{corollary}
	\label{cor: asy_bound_d2}
	Let $p$ be a prime power. Let $r$ be a positive even integer.
  There exists a sequence of linear sum-rank codes of matrix size $2 \times 2$ over $\F_{p}$ with relative sum-rank distance $\delta_{sr} \in (0,1)$
  such that the rate $R_{sr}$ satisfies 
  \begin{align*}
		\frac{r+2}{r}R_{sr} + \frac{r+2}{2} \delta_{sr} \geq 1 - \frac{1}{p^{r}-1},
	\end{align*}

\end{corollary}

For $p=2$ and $r=4$,
Corollary \ref{cor: asy_bound_d2} shows that 
there exists a sequence of 
linear sum-rank codes over ${\bf{F}}_{2}$ of matrix size $2 \times 2$ with relative sum-rank distance $\delta_{sr}$ and rate $R_{sr}$ 
satisfying
$$ \frac{3}{2}R_{sr} + 3 \delta_{sr} \geq \frac{14}{15}. $$
Recall that the GV-like bound with $q=2$ and $n=m=2$
shows that there exists a sequence of linear
sum-rank codes over ${\bf{F}}_{2}$ of matrix size $2 \times 2$ with rate $R_{sr}$ and relative sum-rank distance at least $\delta_{sr}$, where  
$R_{sr}$ and $\delta_{sr}$ satisfy (\ref{eq: asy_GV_bound}).
Figure \ref{fig:p2r4_gv} gives the comparison of Corollary \ref{cor: asy_bound_d2} and
the GV-like bound for linear sum-rank codes over $\F_{2}$ of matrix size $2 \times 2$.
One can see that given $\delta_{sr}$,
the value of $R_{sr}$ given by our bound exceeds the GV-like bound.

For $p=3^2$ and $r=6$,
Corollary \ref{cor: asy_bound_d2} shows that 
there exists a sequence of 
linear sum-rank codes over ${\bf{F}}_{9}$ of matrix size $2 \times 2$ with relative sum-rank distance $\delta_{sr}$ and rate $R_{sr}$ 
satisfying
$$ \frac{4}{3}R_{sr} + 4 \delta_{sr} \geq 1 - \frac{1}{3^{12}-1}. $$
The TVZ-like bound with such $p$ and $m=2$ shows that there exists a sequence of linear
sum-rank codes over ${\bf{F}}_{9}$ of matrix size $2 \times 2$ with rate $R_{sr}$ and relative sum-rank distance at least $\delta_{sr}$, where  
$R_{sr} + \delta_{sr} \leq \frac{1}{4}.$
Figure \ref{fig:p9r6_tvz} gives the 
comparison of Corollary \ref{cor: asy_bound_d2} with the TVZ-like bound for sum-rank codes over $\F_{9}$ of matrix size $2 \times 2$.

Note that the TVZ-like bound is invalid for linear sum-rank codes over $\F_{4}$ of matrix size $2 \times 2$.
Corollary \ref{cor: asy_bound_d2} shows that there exists a sequence of 
linear sum-rank codes over ${\bf{F}}_{4}$ of matrix size $2 \times 2$ with relative sum-rank distance 
$\delta_{sr}$ and rate $R_{sr}$ satisfying $$ \frac{3}{2}R_{sr} + 3 \delta_{sr} \geq \frac{254}{255}, $$ see Figure \ref{fig:p4r4}.

\begin{figure}[htbp]
\centering
\begin{minipage}{0.45\textwidth}
  \centering
  \includegraphics[width=\linewidth]{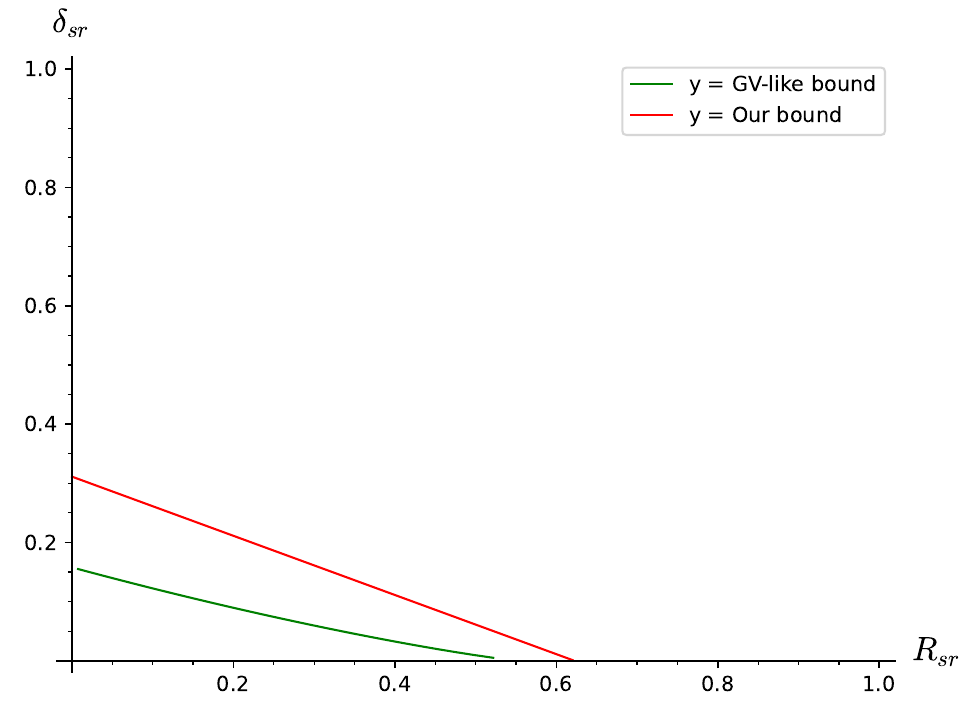}
  \caption{Comparison of Corollary \ref{cor: asy_bound_d2} with the GV-like bound for sum-rank codes of matrix size $2 \times 2$ over $\F_{2}$.}
  \label{fig:p2r4_gv}
\end{minipage}
\hfill
\begin{minipage}{0.45\textwidth}
  \centering
  \includegraphics[width=\linewidth]{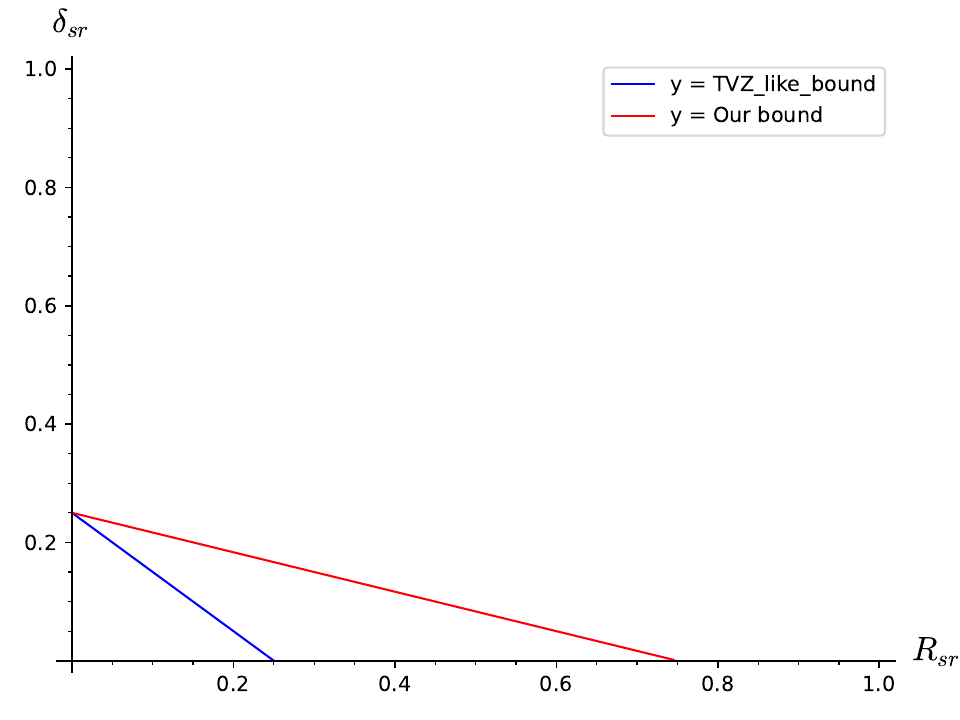}
  \caption{Comparison of Corollary \ref{cor: asy_bound_d2} with the TVZ-like bound for sum-rank codes of matrix size $2 \times 2$ over $\F_{9}$.}
  \label{fig:p9r6_tvz}
\end{minipage}
\end{figure}

  


\begin{figure}[htbp]
    \centering
    \includegraphics[width=0.43\linewidth]{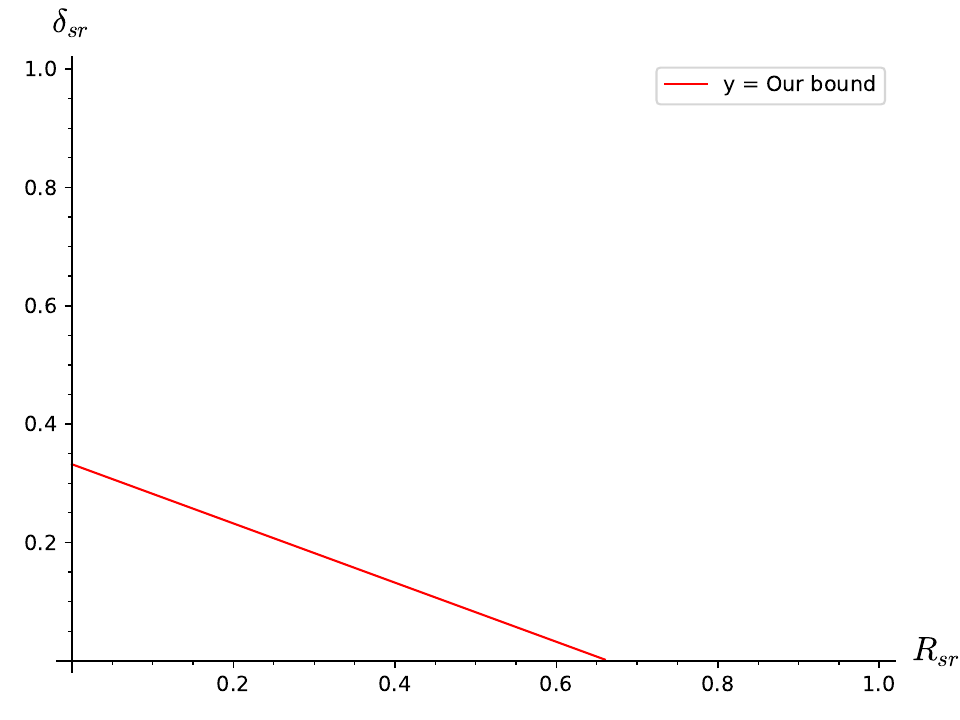}
    \caption{Corollary \ref{cor: asy_bound_d2} with $p=4$ and $r=4$.}
    \label{fig:p4r4}
\end{figure}

\subsection{Concatenated codes from linearized Reed-Solomon codes}

The linearized Reed-Solomon codes constructed in \cite{MP1} are linear sum-rank codes of block size $m \times m$ over $\F_{q}$ with block length $t \leq q-1$.
If the minimum sum-rank distance of a linearized Reed-Solomon code is $d$, where $0 < d \leq mt$, then its dimension is $k=m(mt-d+1)$.
Note that when $t=1$, a linearized Reed-Solomon code is a Gabidulin code.

When the sum-rank metric code ${\bf C}$ in Theorem \ref{thm: asymptotically} is a linearized Reed-Solomon code,
we have the following result.
\begin{theorem}
	\label{thm: from_rs}
	Let $p$ be a prime power. Let $m$, $t$, and $d$  be integers such that 
	$t \leq p-1$, $0 < d \leq mt$, and $m(mt-d+1)$ is an even integer.
  Let $r = m(mt-d+1)/2$.
  There exists a sequence of linear sum-rank codes of matrix size $m \times m$ over $\F_{p}$ with relative sum-rank distance $\delta_{sr} \in (0,1)$
  such that the rate $R_{sr}$ satisfies
  \begin{align*}
		\frac{m^{2}t}{2r}R_{sr} + \frac{mt}{d}\delta_{sr} \geq 1 - \frac{1}{p^{r}-1}.
	\end{align*}
\end{theorem}

Let $m$ be even and $d=\frac{mt}{2}$ in Theorem \ref{thm: from_rs}.
We have the following corollary.
\begin{corollary}
	\label{cor: from_rs}
	Let $p$ be a prime power. Let $m$ be an even integer and let $t$ be an integer with 
	$t \leq p-1$. Let $r = \frac{m^2t}{4} + \frac{m}{2}$.
  There exists a sequence of linear sum-rank codes of matrix size $m \times m$ over $\F_{p}$ with relative sum-rank distance $\delta_{sr} \in (0,1)$
  such that the rate $R_{sr}$ satisfies
  \begin{align}
    \label{eq: cor_from_rs}
		\frac{2mt}{mt+2}R_{sr} + 2\delta_{sr} \geq 1 - \frac{1}{p^{r}-1}.
	\end{align}
	
\end{corollary}

We give a comparison of Corollary \ref{cor: from_rs} and the GV-like bound for sum-rank codes over $\F_{3}$ of matrix size $2 \times 2$,
in Figure \ref{fig:p3m2t2_different_t_with_gv}. Note that by choosing a different value of $t$ in Corollary \ref{cor: from_rs},
the bound given in (\ref{eq: cor_from_rs}) is different. The red line (respectively, blue line) in the figure is given by Corollary \ref{cor: from_rs} with $t=1$ (respectively, $t=2$).
For both choices of $t$, our bound exceeds the GV-like bound for some ranges of $\delta_{sr}$.

The comparison of Corollary \ref{cor: from_rs} and the TVZ-like bound for sum-rank codes over $\F_{4^2}$ of matrix size $4 \times 4$,
is given in Figure \ref{fig:p16t8m4_tvz}, where the parameters in Corollary \ref{cor: from_rs} are given by $p=16$, $m=4$, and $t=8$.



\begin{figure}[htbp]
\centering
\begin{minipage}{0.45\textwidth}
  \centering
  \includegraphics[width=\linewidth]{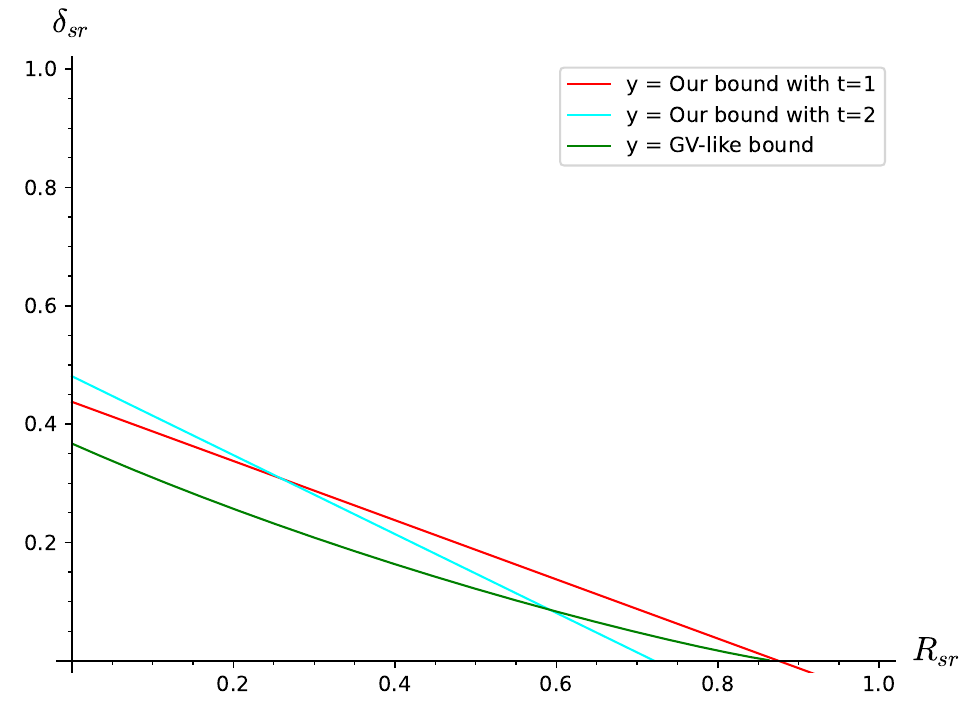}
  \caption{Comparison of the GV-like bound and Corollary \ref{cor: from_rs} for sum-rank codes over $\F_{3}$ of matrix size $2 \times 2$.}
  \label{fig:p3m2t2_different_t_with_gv}
\end{minipage}
\hfill
\begin{minipage}{0.45\textwidth}
  \centering
  \includegraphics[width=\linewidth]{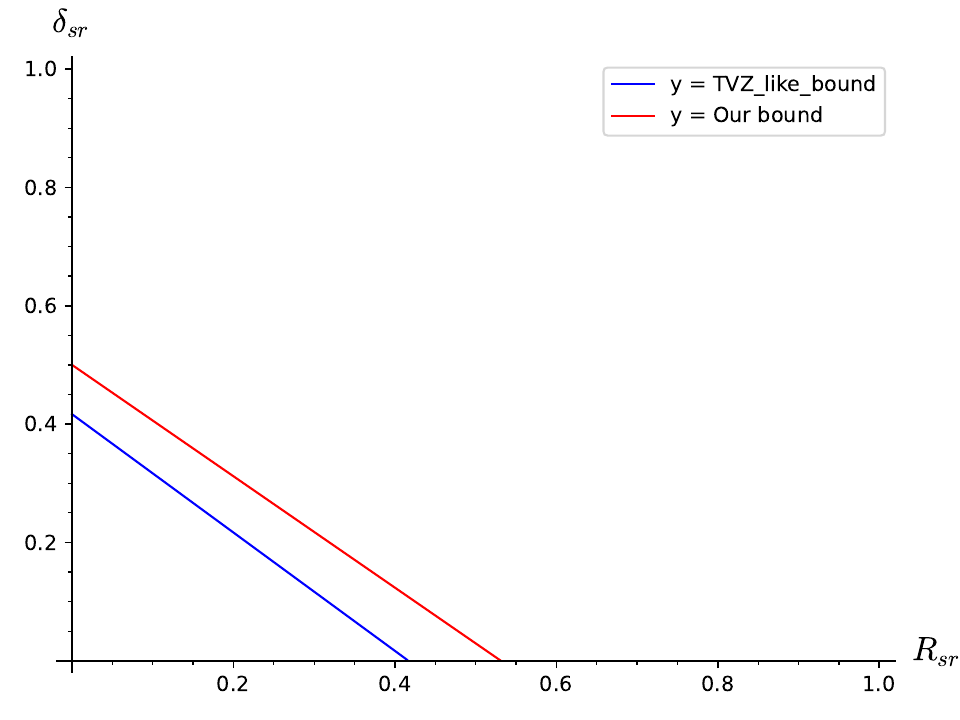}
  \caption{Comparison of the TVZ-like bound and Corollary \ref{cor: from_rs} for sum-rank codes over $\F_{4^2}$ of matrix size $4 \times 4$.}
  \label{fig:p16t8m4_tvz}
\end{minipage}
\end{figure}



Let $m$ be even and let $d=mt$ in Theorem \ref{thm: from_rs}.
We have the following corollary.
\begin{corollary}
	\label{cor: from_rs_max}
	Let $p$ be a prime power. Let $m$ be an even integer and let $t$ be an integer with 
	$t \leq p-1$.
	Let $r = \frac{m}{2}$. 
  There exists a sequence of linear sum-rank codes of matrix size $m \times m$ over $\F_{p}$ with relative sum-rank distance $\delta_{sr} \in (0,1)$
  such that the rate $R_{sr}$ satisfies \begin{align}
		mtR_{sr} + \delta_{sr} \geq 1 - \frac{1}{p^{r}-1}.
	\end{align}
\end{corollary}

The comparison of Corollary \ref{cor: from_rs_max} with the GV-like bound for sum-rank codes of matrix size $2 \times 2$ over $\F_{7}$
is given in Figure \ref{fig:p7t2m2_gv}, where the parameters in Corollary \ref{cor: from_rs_max} are given by $p=7$,$m=2$ and $t=2$.

The comparison of Corollary \ref{cor: from_rs_max} with the TVZ-like bound for sum-rank codes of matrix size $2 \times 2$ over $\F_{5^2}$
is given in Figure \ref{fig:p25t4m2_tvz}, where the parameters in Corollary \ref{cor: from_rs_max} are given by $p=5^2$, $m=2$ and $t=4$.

In addition, we also give a comparison of Corollary \ref{cor: from_rs_max}, the GV-like bound, and the TVZ-like bound
for sum-rank codes of matrix size $2 \times 2$ over $\F_{7^2}$ in Figure \ref{fig:p49r2m2t2_multiple_bound_maximal}.
The red line (respectively, blue line) in the figure is given by Corollary \ref{cor: from_rs_max} with $t=1$ (respectively, $t=2$).


\begin{figure}[htbp]
\centering
\begin{minipage}{0.45\textwidth}
  \centering
  \includegraphics[width=\linewidth]{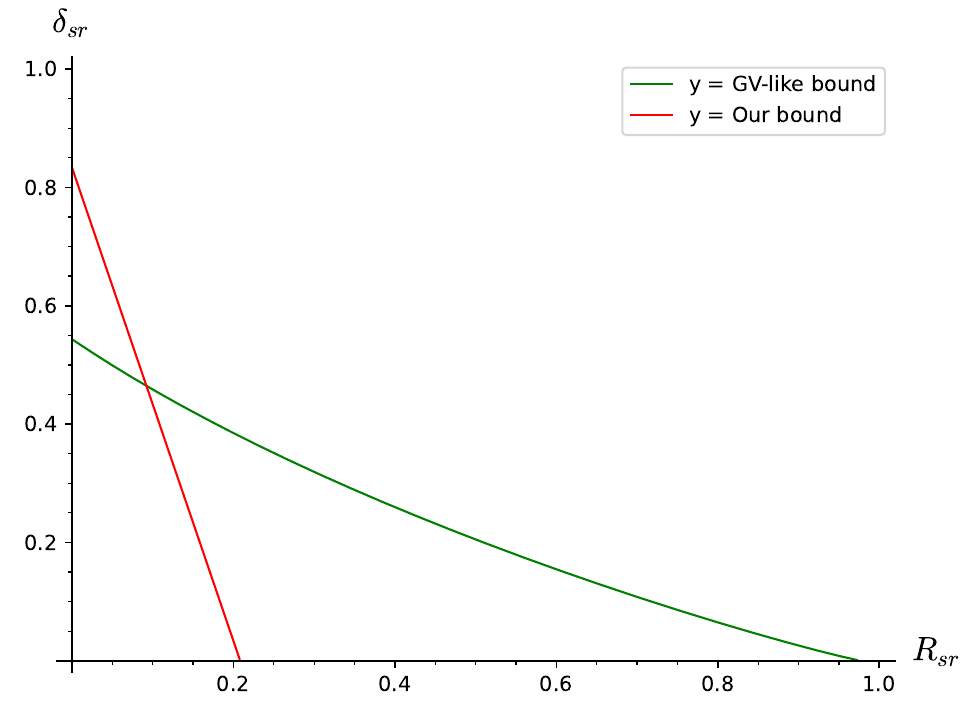}
  \caption{Comparison of Corollary \ref{cor: from_rs_max} with the GV-like bound for sum-rank codes of matrix size $2 \times 2$ over $\F_{7}$.}
  \label{fig:p7t2m2_gv}
\end{minipage}
\hfill
\begin{minipage}{0.45\textwidth}
  \centering
  \includegraphics[width=\linewidth]{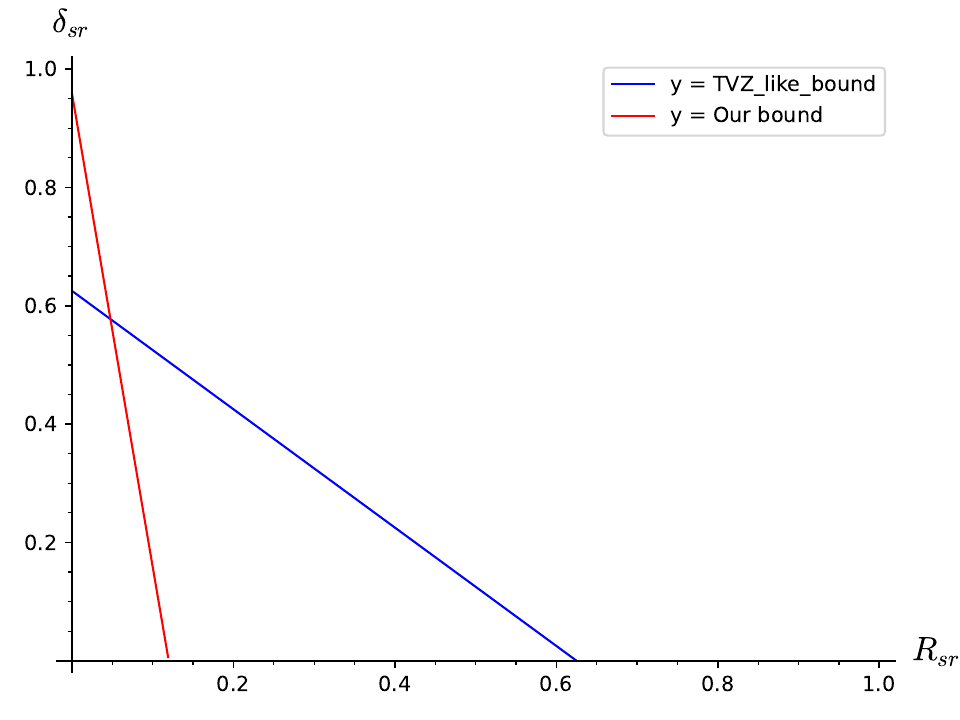}
  \caption{Comparison of Corollary \ref{cor: from_rs_max} with the TVZ-like bound for sum-rank codes of matrix size $2 \times 2$ over $\F_{5^2}$.}
  \label{fig:p25t4m2_tvz}
\end{minipage}
\end{figure}

  


\begin{figure}[htbp]
    \centering
    \includegraphics[width=0.43\linewidth]{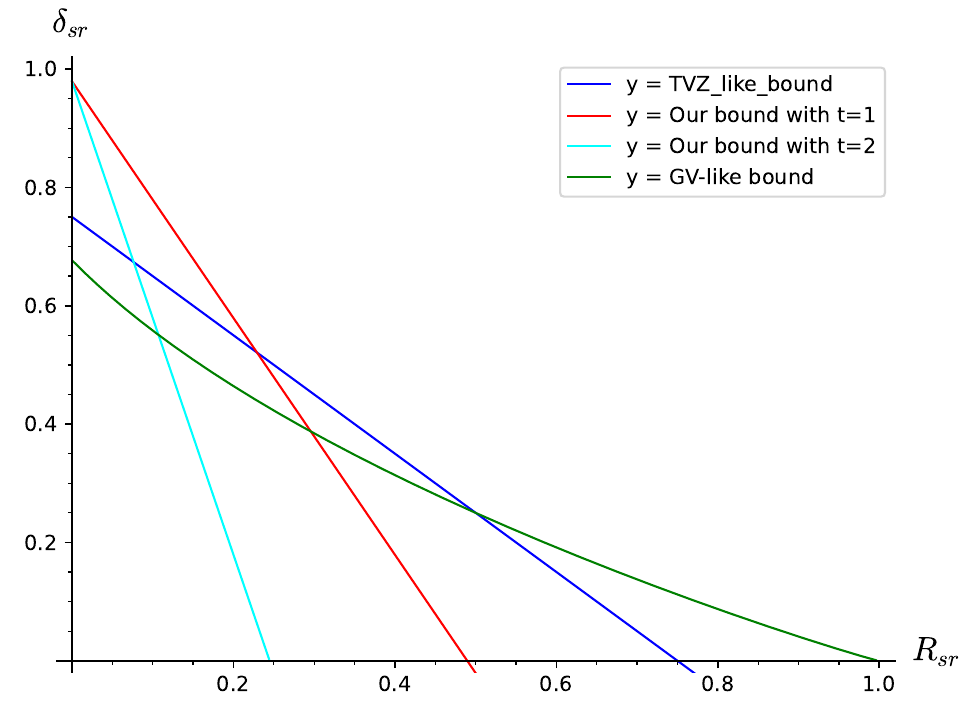}
    \caption{Comparison of the GV-like bound, the TVZ-like bound, and Corollary \ref{cor: from_rs_max} for sum-rank codes of matrix size $2 \times 2$ over $\F_{7^2}$.}
    \label{fig:p49r2m2t2_multiple_bound_maximal}
\end{figure}


\section{Explicit constructions}
\label{sec: construct}
Based on Theorem \ref{thm: CSR}, we have the following explicit construction for linear sum-rank codes.
\begin{corollary}
	\label{cor: explicit_construction}
	Let $q$ be a prime power, and let $m$ and $d$ be integers with $0 < d \leq m$.
	Let ${\bf C}_{1}$ be an $[n, k_{1}, d_{1}]_{q^{k_{2}}}$ linear code,
	where $k_{2} = m(m-d+1)$.
	Let ${\bf C}_{2}$ be a maximum rank metric code in $\F_{q}^{(m,m)}$ with minimum rank distance $d$.
	Then there is a sum-rank code ${\bf C}$ over $\F_{q}$ with block length $n$ and matrix size $m \times m$.
	Its minimum sum-rank distance $d_{sr}$ is $d_{1} \cdot d$, and its dimension is $k_{1} \cdot k_{2}$.
\end{corollary}

When $q=2$, $n=15$, $m=2$, and $d=1$ in Corollary \ref{cor: explicit_construction},
we have $k_{2}=4$. Let ${\bf C}_{1}$ be an $[n, n-d_{1}+1, d_{1}]_{16}$ Reed-Solomon code,
where $d_{1} \in [4, 15]$. 
Then the minimum sum-rank distance and the dimension of ${\bf C}$ over $\F_{2}$
are given in Table \ref{table1_block_length15}. We compare them with \cite[Table IV]{MP21}.
One can see that the dimensions of the codes obtained by our construction are larger than those given in \cite{MP21}.

When $q=3$, $n=31$, $m=2$, and $d=1$ in Corollary \ref{cor: explicit_construction},
we have $k_{2}=4$. Let ${\bf C}_{1}$ be an $[n, n-d_{1}+1, d_{1}]_{3^4}$ Reed-Solomon code,
where $d_{1} \in [4, 31]$. 
Then the minimum sum-rank distance and the dimension of ${\bf C}$ over $\F_{3}$
are given in Table \ref{table1_block_length31}. We compare them with \cite[Table VII]{Chen}.
The numbers in bold face in the table mean that the dimensions of our constructed codes are larger.

\begin{table}[htbp]
\centering
\caption{Sum-rank codes with block length $15$ and block size $2 \times 2$ over $\F_{2}$.}
\label{table1_block_length15}
\begin{tabular}{|l|l|l|l|l|}

\hline $d_{sr}$ &  dimension  & Table IV \cite{MP21} & Singleton \\
\hline $4$ &  	$2 \cdot 24$  & $2 \cdot 20$ & $2 \cdot 27$ \\
\hline $5$ &   $2 \cdot 22$  & $2 \cdot 18$ & $2 \cdot 26$ \\
\hline $6$ &   $2 \cdot 20$  & $2 \cdot 16$ & $2 \cdot 25$ \\
\hline $7$ &   $2 \cdot 18$ & $2 \cdot 14$ & $2 \cdot 24$ \\
\hline $8$ & 	$2 \cdot 16$  & $2 \cdot 10$ & $2 \cdot 23$ \\
\hline $9$ & 	$2 \cdot 14$  & $2 \cdot 8$ & $2 \cdot 22$ \\
\hline $10$ & 	$2 \cdot 12$ & $2 \cdot 8$ & $2 \cdot 21$ \\
\hline $11$ & 	$2 \cdot 10$  & $2 \cdot 6$ & $2 \cdot 20$ \\
\hline $12$ & 	$2 \cdot 8$  & $2 \cdot 4$ & $2 \cdot 19$ \\
\hline $13$ & 	$2 \cdot 6$  & none & $2 \cdot 18$ \\
\hline $14$ & 	$2 \cdot 4$  & $2 \cdot 2$ & $2 \cdot 17$ \\
\hline $15$ & 	$2 \cdot 2$  & none & $2 \cdot 16$ \\
\hline
\end{tabular}
\end{table}

\begin{table}[htbp]
\centering
\caption{Sum-rank codes with block length $31$ and block size $2 \times 2$ over $\F_{3}$.}
\label{table1_block_length31}
\begin{tabular}{|l|l|l|l|}
\hline $d_{s r}$ &	Dimension	& Table VII \cite{Chen} & Singleton \\
\hline 4 & 	$2 \cdot 56$ &	$2 \cdot 57$ & $2 \cdot 59$ \\
\hline 5 & 	$2 \cdot \bf{54}$ & $2 \cdot 53$ & $2 \cdot 58$ \\
\hline 6 & 	$2 \cdot 52$  &			$2 \cdot 52$ & $2 \cdot 57$ \\
\hline 7 & 	$2 \cdot \bf{50}$ &			$2 \cdot 49$ & $2 \cdot 56$ \\
\hline 8 & 	$2 \cdot \bf{48}$  &			$2 \cdot 47$ & $2 \cdot 55$ \\
\hline 9 & 	$2 \cdot \bf{46}$ &			$2 \cdot 44$ & $2 \cdot 54$ \\
\hline 10 & $2 \cdot \bf{44}$ &			$2 \cdot 42$ & $2 \cdot 53$ \\
\hline 11 & $2 \cdot \bf{42}$ &			$2 \cdot 40$ & $2 \cdot 52$ \\
\hline 12 & $2 \cdot \bf{40}$ &			$2 \cdot 39$ & $2 \cdot 51$ \\
\hline 13 & $2 \cdot \bf{38}$ &			$2 \cdot 36$ & $2 \cdot 50$ \\
\hline 14 & $2 \cdot \bf{36}$ &			$2 \cdot 35$ & $2 \cdot 49$ \\
\hline 15 & $2 \cdot \bf{34}$ &			$2 \cdot 32$ & $2 \cdot 48$ \\
\hline 16 & $2 \cdot \bf{32}$ &			$2 \cdot 31$ & $2 \cdot 47$ \\
\hline 17 & $2 \cdot \bf{30}$ &			$2 \cdot 29$ & $2 \cdot 46$ \\
\hline 18 & $2 \cdot 28$ &			$2 \cdot 28$ & $2 \cdot 45$ \\
\hline 19 & $2 \cdot \bf{26}$ &			$2 \cdot 25$ & $2 \cdot 44$ \\
\hline 20 & $2 \cdot 24$ &			$2 \cdot 24$ & $2 \cdot 43$ \\
\hline 21 & $2 \cdot 22$ &			$2 \cdot 22$ & $2 \cdot 42$ \\
\hline 22 & $2 \cdot 20$ &			$2 \cdot 21$ & $2 \cdot 41$ \\
\hline 23 & $2 \cdot 18$ &			$2 \cdot 20$ & $2 \cdot 40$ \\
\hline 24 & $2 \cdot 16$ &			$2 \cdot 19$ & $2 \cdot 39$ \\
\hline 25 & $2 \cdot 14$ &			$2 \cdot 18$ & $2 \cdot 38$ \\
\hline 26 & $2 \cdot 12$ &			$2 \cdot 17$ & $2 \cdot 37$ \\
\hline 27 & $2 \cdot 10$  &			$2 \cdot 15$ & $2 \cdot 36$ \\
\hline 28 & $2 \cdot 8$ &			$2 \cdot 14$ & $2 \cdot 35$ \\
\hline 29 & $2 \cdot 6$ &			$2 \cdot 13$ & $2 \cdot 34$ \\
\hline 30 & $2 \cdot 4$ &			$2 \cdot 13$ & $2 \cdot 33$ \\
\hline
\end{tabular}
\end{table}

\section{Conclusion}
\label{sec: conclusion}
In this paper, we introduce concatenated sum-rank codes. Then we show that there exists a sequence of concatenated sum-rank codes exceeding the asymptotic
Tsfasman-Vl\u{a}du\c{t}-Zink-like bound and the asymptotic Gilbert-Varshamov-like bound.
In addition, we also give explicit constructions of concatenated sum-rank codes. We show that for some parameter regimes,
the dimensions of our constructed codes are better than those given in \cite{MP21} and \cite{Chen}.

\end{document}